\documentclass[12pt]{article}
\usepackage[T1]{fontenc}
\usepackage[latin9]{inputenc}
\usepackage{geometry}
\usepackage[english]{babel}
\usepackage{array}
\usepackage{float}
\usepackage{mathtools}
\usepackage{url}
\usepackage{amsmath}
\usepackage{amsthm}
\usepackage{amssymb}
\usepackage{graphicx}
\usepackage{setspace}
\usepackage[round]{natbib}
\usepackage[pdfusetitle,
 bookmarks=true,bookmarksnumbered=false,bookmarksopen=false,
 breaklinks=false,pdfborder={0 0 0},pdfborderstyle={},backref=false,colorlinks=false]
 {hyperref}

\makeatletter

\providecommand{\tabularnewline}{\\}

\theoremstyle{plain}
\newtheorem{thm}{\protect\theoremname}[section]
\theoremstyle{remark}
\newtheorem{rem}[thm]{\protect\remarkname}
\theoremstyle{plain}
\newtheorem{prop}[thm]{\protect\propositionname}
\theoremstyle{plain}
\newtheorem{lem}[thm]{\protect\lemmaname}

\usepackage{hyperref}

\makeatother

\providecommand{\lemmaname}{Lemma}
\providecommand{\propositionname}{Proposition}
\providecommand{\remarkname}{Remark}
\providecommand{\theoremname}{Theorem}

\begin{document}
\title{Failure of Fourier pricing techniques to approximate the Greeks}
\author{Tobias Behrens\thanks{Carl von Ossietzky Universität, Institut für Mathematik, 26129 Oldenburg,
Germany.}, Gero Junike\thanks{Corresponding author. Carl von Ossietzky Universität, Institut für
Mathematik, 26129 Oldenburg, Germany, ORCID: 0000-0001-8686-2661,
E-mail: gero.junike@uol.de}, Wim Schoutens\thanks{KU Leuven, Department of Mathematics, Celestijnenlaan 200B, 3001 Leuven,
Belgium, ORCID: 0000-0001-8510-1510, E-mail: wim.schoutens@kuleuven.be}}
\maketitle
\begin{abstract}
The Greeks Delta and Gamma of plain vanilla options play a fundamental
role in finance, e.g., in hedging or risk management. These Greeks
are approximated in many models such as the widely used Variance Gamma
model by Fourier techniques such as the Carr-Madan formula, the COS
method or the Lewis formula. However, for some realistic market parameters,
we show empirically that these three Fourier methods completely fail
to approximate the Greeks. As an application we show that the Delta-Gamma
VaR is severely underestimated in realistic market environments. As
a solution, we propose to use finite differences instead to obtain
the Greeks.\\
 \textbf{Keywords: }Fourier pricing, Carr-Madan formula, COS method,
Lewis formula, Greeks, Gamma\\
 \textbf{Mathematics Subject Classification} 60E10 · 91G20 · 91G60\\
 \textbf{JEL Classification }C63 · G13 
\end{abstract}

\section{Introduction }

Financial contracts such as European options are widely used by financial
institutions. The derivatives or sensitivities of options are important
parameters for hedging and risk management purposes and are known
as Greeks. The most important Greeks are Delta and Gamma, which are
the first and second derivatives of a European option with respect
to the current price of the underlying asset.

For some simple pricing models such as the Black-Scholes model, there
are explicit formulas for the price, the Delta and the Gamma of different
European options. However, for many advanced stock price models such
as the Variance Gamma (VG) model\footnote{The VG model is implemented in Quantlib, a widely used software package
in the financial industry, see https://www.quantlib.org/}, see \citet{madan1998variance}, there are no explicit formulas for
the price and the Greeks, even for plain vanilla options. However,
since the characteristic function of the log returns in the VG model
is explicitly given, fast Fourier techniques can be used to obtain
prices and Greeks numerically efficiently.

\citet{heston1993closed} was one of the first to introduce Fourier
pricing techniques. Since then, many other Fourier pricing techniques
have been considered, in particular the Carr-Madan formula, the Lewis
formula, the COS method, the CONV method and pricing with wavelets,
see \citet{carr1999option}, \citet{lewis2001simple}, \citet{eberlein2010analysis},
\citet{fang2009novel}, \citet{junike2022precise}, \citet{lord2008fast},\\
\citet{ortiz2013robust} and \citet{ortiz2016highly}. Among many 
others, the following researchers apply Fourier techniques to obtain
the Greeks:\\ \citet{takahashi2009efficient}, \citet{eberlein2016computation}
and \citet{leitao2018data}. Formally, one often has to interchange
integration and differentiation to obtain formulas for the Greeks.
However, this is not always allowed!

In this paper, we make the following contributions: The characteristic
function of the log returns in the VG model has Pareto heavy tails
and the density of the log returns is unbounded, in particular, it
is not differentiable, see \citet{kuchler2008shapes}. We introduce
the \emph{mixture exponential }(ME) model, which is a generalization
of the Laplace model, see \citet{madan2016adapted}. Similar to the
VG model, the characteristic function in the ME model has Pareto heavy
tails and the density contains a jump. However, the ME model is easier
to understand and to work with because closed form expressions for
the characteristic function and the density of the log returns; the
price, the Delta and the Gamma of simple options are explicitly given.

We summarize three Fourier pricing techniques: The Carr-Madan formula,
the COS method and the Lewis formula, and we recall sufficient conditions
from the literature guaranteeing that these three Fourier techniques
can approximate the Greeks Delta and Gamma. We show that the VG and
the ME models do not satisfy these conditions for some (realistic)
model parameters.

Since the Delta in the ME model has a (single) jump, it is understandable
that Fourier pricing techniques locally fail to approximate the Gamma.
However, we show empirically that they also \emph{globally} fail to
approximate the Gamma, even far away from the jump. This is surprising.
For the Gamma, the three Fourier techniques produce only meaningless
results in the VG and the ME models for call options and digital put
options. As a simple solution, we propose to use finite differences
instead of Fourier techniques to obtain the Greeks Delta and Gamma,
compare with \citep[Sec. 7]{glasserman2004monte}. Finite differences
are simple to implement. They work very well and are consistent with
the closed form solution in the ME model. 

As an application, we estimate the value at risk (VaR) of a digital
put option over a time horizon of one-day in the VG and the ME models
using realistic market parameters and three different methods: i)
a full Monte Carlo simulation, ii) Delta-Gamma VaR where Delta and
Gamma are obtained by finite differences, and iii) Delta-Gamma VaR
where Delta and Gamma are obtained by the three Fourier pricing techniques.
We observe that the Delta-Gamma VaR using finite differences is very
close to the ground truth, i.e., the VaR based on the full Monte Carlo
simulation. However, the Delta-Gamma VaR using Fourier pricing techniques
underestimates the VaR by a large factor, in some realistic experiments
by a factor of up to 2000.

In Section \ref{sec:Stock-price-models} we introduce the VG and the
ME models and the Carr-Madan formula, the COS method and the Lewis
formula. In Section \ref{sec:Numerical-experiments}, we perform some
numerical experiments including an application in risk management.
Section \ref{sec:Conclusion} concludes. 

\section{\protect\label{sec:Stock-price-models}Stock price models and Fourier
Pricing techniques}

Let $(\Omega,\mathcal{F},P)$ be a probability space. By $P$ we denote
the physical measure and by $Q$ we denote a risk-neutral measure.
We consider an investor with time horizon $T>0$. Inspired by \citet{corcuera2009implied},
we assume that there is a financial market consisting of a stock with
price $S_{0}>0$ today and (random) price 
\[
S_{T}:=S_{0}\bar{S}_{T}:=S_{0}\exp(rT+m+X_{T}),
\]
at the end of the time horizon, where $r$ describes some interest
rates, $m$ is a mean-correcting term putting us in a risk-neutral
setting and $X_{T}$ is a random variable. We consider two models:
the Variance Gamma (VG) model and the mixture exponential (ME) model.

\emph{VG model:} $X_{T}$ follows a Variance Gamma distribution with
parameters $\nu>0$, $\sigma>0$ and $\theta\in\mathbb{R}$, see \citet{madan1998variance}.
The mean-correcting term and the characteristic function of $\log(\bar{S}_{T})$
are given by 
\[
m=\frac{T}{\nu}\log\big(1-\theta\nu-\frac{\sigma^{2}\nu}{2}\big)\quad\text{and}\quad\varphi_{VG}(u)=\frac{\exp\left(iu(rT+m)\right)}{\big(1-iu\theta\nu+\frac{\sigma^{2}u^{2}\nu}{2}\big)^{\frac{T}{\nu}}}.
\]
The density of $X_{T}$ is unbounded if $T<\frac{\nu}{2}$, see \citet{kuchler2008shapes}.

\emph{ME model:} The ME model is a very simple model. Similar to the
VG model, the density of $X_{T}$ is usually not differentiable: Let
$\lambda,\eta>0$ and $\mu\in\mathbb{R}$. Consider the density 
\begin{equation}
f(x)=\begin{cases}
\frac{\eta}{2}\exp(\eta(x-\mu)) & ,x<\mu\\
\frac{\lambda}{2}\exp(-\lambda(x-\mu)) & ,x\geq\mu\text{.}
\end{cases}\label{eq:f}
\end{equation}
The Laplace distribution, see \citet{madan2016adapted}, is included
as a special case for $\lambda=\eta$. The density $f$ is not continuous
if $\lambda\neq\eta$. We call a random variable with density $f$,
ME($\mu$, $\eta$, $\lambda$) distributed. In the ME model with
parameters $\lambda>0$ and $\eta>0$, we assume that $X_{T}$ is
ME(0, $\frac{\eta}{\sqrt{T}}$, $\frac{\lambda}{\sqrt{T}}$) distributed.
Unlike the Laplace model, the ME model can model gains and losses
of different shapes. Closed form expressions for the characteristic
function $\varphi_{ME}$ of $\log(\bar{S}_{T})$, and the price, Delta
and Gamma of a digital put option can be obtained straightforwardly
and are provided in Appendix \ref{sec:ME}.

We assume that there is a European option with maturity $T>0$. In
particular, we consider plain vanilla call options with payoff $(S_{T}-K)^{+}$
and digital put options with payoff $1_{\{K<S_{T}\}}$, where $K>0$
is the strike.

We denote by $\varphi$ the characteristic function of $\log(\bar{S}_{T})$
and by $f$ the density of $\log(\bar{S}_{T})$. For many models in
mathematical finance, $\varphi$ is given analytically, but $f$ is
not explicitly known. The Delta (Gamma) of an option is the first
(second) partial derivative of the price of the option with respect
to $S_{0}$.

Next, we briefly summarize three widely used Fourier pricing techniques,
namely the Carr-Madan (CM) formula, the COS method and the Lewis formula.
Given an expression for the price, formulas for the Greeks Delta and
Gamma can be easily obtained for the three methods by interchanging
integration and differentiation, which is formally justified in Lemma
\ref{lem:Grubb-CM}.

\emph{Carr-Madan formula: }Let $\alpha>0$ be a damping factor such
that $E\big[S_{T}^{1+\alpha}\big]<\infty$. The Carr-Madan formula
for the price of a call option is given by 
\begin{equation}
\pi_{CM}^{\text{call}}=\frac{e^{-\alpha\log(K)}}{\pi}e^{-rT}\int_{0}^{\infty}\underbrace{\Re\bigg\{ e^{-iv\log(K)}\frac{\varphi(v-i(\alpha+1))e^{(\alpha+1+iv)\log(S_{0})}}{\alpha^{2}+\alpha-v^{2}+i(2\alpha+1)v}\bigg\}}_{=:h(v,S_{0})}dv\text{,}\label{eq:CM}
\end{equation}
see \citet{carr1999option}. Formulas for the Greeks Delta and Gamma
are given in \citet[Sec. 3.1.2]{madan2016applied}. To derive these
formulas, integration and differentiation must be interchanged. Sufficient
conditions are provided in Lemma \ref{lem:Grubb-CM}, in particular,
$v\mapsto\left|\frac{\partial h(v,S_{0})}{\partial S_{0}}\right|$
and $v\mapsto\left|\frac{\partial^{2}h(v,S_{0})}{\partial S_{0}^{2}}\right|$
have to be integrable. Similarly, the price of a digital put option
is given by 
\begin{equation}
\pi_{CM}^{\text{digital}}=-\frac{e^{\alpha\log(K)}}{\pi}e^{-rT}\int_{0}^{\infty}\Re\bigg\{ e^{-iv\log(K)}\frac{\varphi(v+i\alpha)e^{(-\alpha+iv)\log(S_{0})}}{iv-\alpha}\bigg\} dv\text{.}\label{eq:cm_digi}
\end{equation}

\emph{COS method: }The COS method goes back to \citet{fang2009novel}.
The main idea of the COS method is to truncate the (unknown) density
$f$ to a finite interval and to approximate the truncated density
by a Fourier cosine expansion. There is a clever trick to obtain the
cosine coefficients for $f$ in a very fast and robust way using $\varphi$.
Sufficient conditions for approximating the price of a call option
or a digital put option and the Greeks Delta and Gamma are given in
\citet[Thm. 5.2]{junike2024number}. In particular, the density $f$
must be continuously differentiable.

\emph{Lewis formula:} The Lewis formula expresses the price of an
option by 
\begin{equation}
\pi_{Lewis}=\frac{e^{-rT}}{\pi}\int_{0}^{\infty}\underbrace{\Re\big\{ S_{0}^{\alpha-iv}\varphi(-i\alpha-v)\hat{w}(v+i\alpha)\big\}}_{=:g(v,S_{0})}dv,\label{eq:Lewis}
\end{equation}
where $\alpha\in\mathbb{R}$ is a suitable damping factor and $\hat{w}$
is the Fourier transform of the payoff function, e.g., 
\begin{equation}
\hat{w}_{\text{call}}(z)=\frac{K^{1+iz}}{iz(1+iz)},\quad\hat{w}_{\text{digital}}(z)=\frac{K^{iz}}{iz},\label{eq:w_hat}
\end{equation}
where the range for the imaginary part of $z$ is $(1,\infty)$ and
$(-\infty,0)$, respectively. Sufficient conditions for obtaining
prices and Greeks by the Lewis formula are based on Lemma \ref{lem:Grubb-CM}
and are also given in \citet[Sec. 4]{eberlein2010analysis}. In particular,
$v\mapsto\left|\frac{\partial g(v,S_{0})}{\partial S_{0}}\right|$
and $v\mapsto\left|\frac{\partial^{2}g(v,S_{0})}{\partial S_{0}^{2}}\right|$
have to be integrable. How to choose $\alpha$ can be found in \citet{bayer2022optimal}.

\section{\protect\label{sec:Numerical-experiments}Numerical experiments}

We use the following parameters for the VG and the ME models in order
to perform numerical experiments: $K=0.75$, $T=\frac{1}{12}$, $r=0$.
For the VG model we use the parameters
\begin{equation}
\sigma=0.13,\quad\theta=0,\quad\nu=0.4.\label{eq:paramVG}
\end{equation}
For the ME model, we set
\begin{equation}
\eta=1,\quad\lambda=2.\label{eq:paramME}
\end{equation}
These are realistic market parameters, especially considering the
VG model, compare with \citet[Tab. II]{madan1998variance}. For these
parameters, the conditions mentioned in Section \ref{sec:Stock-price-models}
of the three Fourier techniques to approximate the Delta and Gamma
\emph{are not met}, see Remark \ref{rem:notMet}\emph{.} As shown
in Figure \ref{fig:ME_model}, the Carr-Madan formula, the COS method
and the Lewis formula price a call option or digital put option correctly
but do not approximate the Gamma in any sense. A simple solution would
be to use finite differences to obtain the Greeks instead of Fourier
pricing techniques. Remark \ref{rem:Finite-differences-step} discusses
how to set the step size for applying finite differences. Implementation
details are given in Remark \ref{rem:Implementation-details:-For}. 
\begin{rem}
\label{rem:notMet}Note that the $|\hat{w}_{call}|$ and $|\hat{w}_{digi}|$
in (\ref{eq:w_hat}) are bounded. Suppose $T<\frac{\nu}{2}$ in the
case of the VG model and $\lambda\neq\eta$ in the case of the ME
model. We have that 
\[
|\varphi_{VG}(u)|=O\left(u^{-\frac{2T}{\nu}}\right),\quad u\to\infty\quad\text{and}\quad|\varphi_{ME}(u)|=O\left(u^{-1}\right),\quad u\to\infty.
\]
Thus, it is easy to see that the sufficient conditions of the Carr-Madan
formula and the Lewis formula for obtaining the Gamma of a call option
or digital put option in the VG or the ME models are not satisfied.
Furthermore, the density of $\log(S_{T})$ in the VG model is unbounded
if $T<\frac{\nu}{2}$, see \citet{kuchler2008shapes}. Thus, the density
in the VG model is not differentiable and the sufficient conditions
for the COS method to approximate the Gamma are also not satisfied.
The density is also not differentiable in the ME model if $\eta\neq\lambda$. 
\end{rem}

\begin{figure}[H]
\begin{centering}
\includegraphics[scale=0.8]{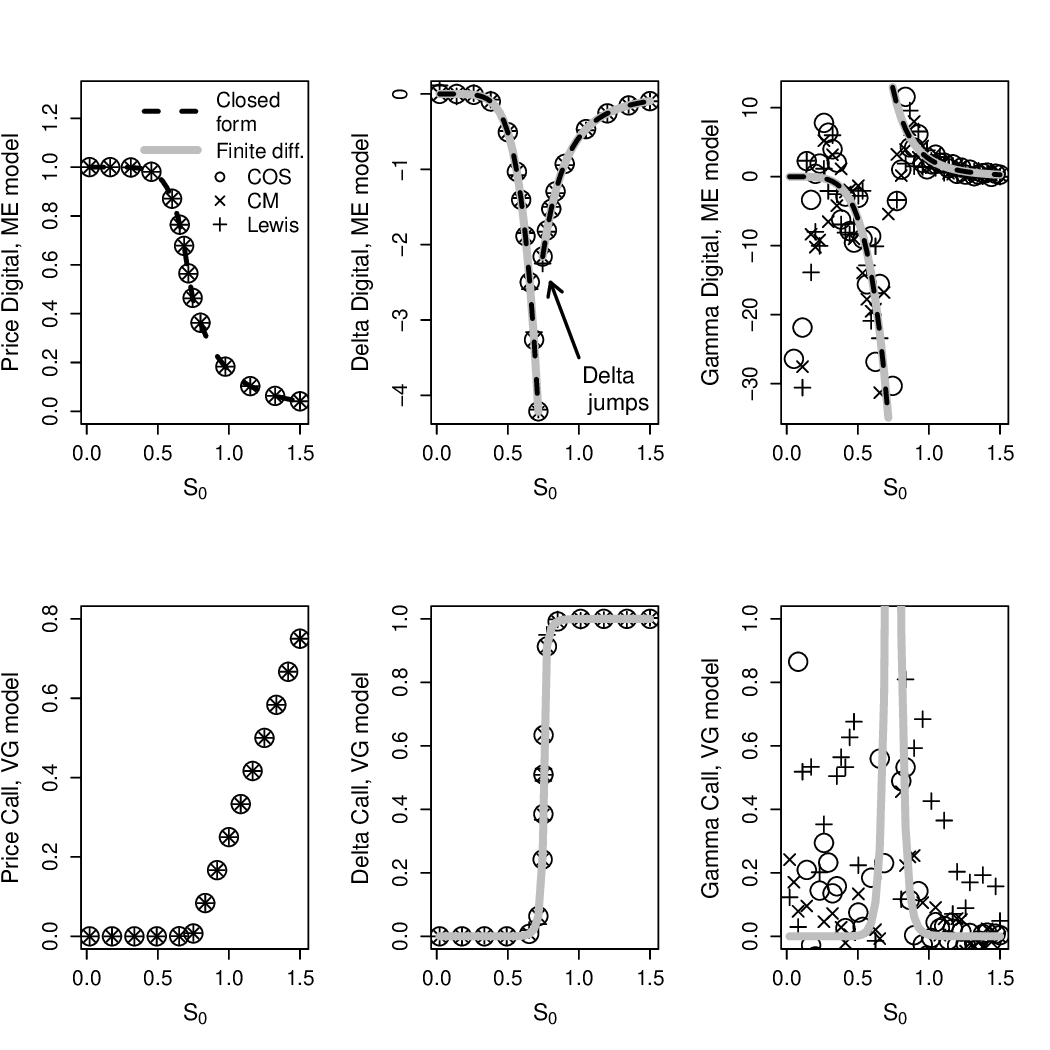}
\par\end{centering}
\centering{}\caption{\protect\label{fig:ME_model}Prices and Greeks of a call option and
a digital put option using the Carr-Madan (CM) formula, the COS method,
the Lewis formula and finite differences under the ME and the VG models.}
\end{figure}

\begin{rem}
\label{rem:Finite-differences-step}There is a rich literature on
how to choose the step size for applying central finite differences,
see, e.g., \citet[Sec. 6.4]{martins2021engineering}; \citet[Sec. 2]{OwolabiKoladeM2019Nmff};
\citet[Sec. 5.1]{sauer2013}. In particular, in \citet[Sec. 5.1.2]{sauer2013},
the optimal step size is given by the cube root of the machine precision,
assuming that the error between the true price and the floating point
version of the price is dominated by the machine precision, which
is approximately equal to $\varepsilon_{\text{mach}}=10^{-15}$. However,
in our case, the error is dominated by the error of the Fourier estimator.
Following the same arguments as in \citet[Sec. 5.1.2]{sauer2013},
the optimal step size for applying finite differences should be of
the same order of magnitude as the cube root of the error introduced
by the Fourier estimator. For example, if the Fourier-pricing error
is of order $10^{-6}$, we recommend choosing the step-size $h=0.01$.
For example, for the COS method, the error introduced by the Fourier
estimator is explicitly know a priori, see \citet{junike2022precise},
\citet{junike2024number}.
\end{rem}

\begin{rem}
\label{rem:Implementation-details:-For}Implementation details: For
the Carr-Madan formula, the integrals in (\ref{eq:CM}) and (\ref{eq:cm_digi})
are solved numerically by Simpson's rule and we set the number of
grid points to $N=2^{17}$, the damping factor to $\alpha=0.1$ and
the Fourier truncation range to $4800$. For the COS method, we use
a truncation range of size $60$ and set the number of terms to $10^{5}$.
For the Lewis formula, we use a damping factor $\alpha=1.1$ for the
call option and $\alpha=-0.1$ for the digital put option. We use
a number of terms $N=2000$ to apply Gauss-Laguerre quadrature to
approximate the integral in (\ref{eq:Lewis}). We implement central
finite differences with step size $h=0.01$. The R-code of the three
methods can be found in \url{https://github.com/GeroJunike/Fourier_pricing} 
\end{rem}

\subsection{Application to Delta-Gamma VaR}

We consider a portfolio $C$ consisting of a single digital put option
with strike $K=0.75$ and maturity $T=\frac{1}{12}$ and a loan,
i.e,
\[
C:=\pi(S_{t})-\pi(S_{0}),
\]
where $S_{0}$ denotes the current stock price and $S_{t}$ denotes
the future stock price over a time horizon of one day. $\pi(s)$ is
the price of the option when the current asset price is $s$ and we
assume that $\pi$ is twice differentiable in $s$. As before, we
set $r=0$.

The one-day VaR at a fixed level $p\in(0,1)$ for the financial position
$C$ is defined as
\[
\text{VaR}_{p}(C)\coloneqq\inf\{x|P(C+x<0)\leq p\},
\]
Usually we use $p=0.01.$ ``This means that the $\text{VaR}_{p}(C)$
is the smallest amount of capital which, if added to $C$ and invested
in the risk-free asset, keeps the probability of a negative outcome
below the level $p$'', see \citet[p. 231]{follmer2011stochastic}.

In principle, the VaR can be estimated by Monte Carlo simulation:
simulate the stock price at the one-day horizon $N$ times and calculate
the option price for each scenario\footnote{As noted in \citet{duffie2001analytical}, banks often ignore any
drift under the physical measure for the purposes of short-horizon
VaR. We follow this approach and simulate $S_{t}$ under $Q$, i.e.,
assuming a risk-neutral drift term. This approach is valid if the
model parameters under the physical measure $P$ and the risk-neutral
measure $Q$ differ only in the drift term, compare with \citet[Sec. 6.2.2]{schoutens2003levy}. }. The full Monte Carlo VaR is then obtained from the empirical $0.01$-quantile
of the $N$ option prices, see \citet[Sec. 22.2]{HullJohn2021Ofao}.
This approach is computationally expensive when $N$ is large, since
a Fourier-pricing technique must be applied $N$ times to price the
option under all scenarios.

The Delta-Gamma VaR uses a second order approximation of the price
$\pi$ and is much faster to compute. It is defined, for example,
in \citet{duffie2001analytical} and it is applicable when the one-day
changes in the stock price are small. Using Taylor's theorem, we write
the value of the portfolio $C$ as a function of $S_{t}$ by
\begin{equation}
C=\pi(S_{t})-\pi(S_{0})\approx\Delta(S_{0})\times(S_{t}-S_{0})+\frac{1}{2}\Gamma(S_{0})\times(S_{t}-S_{0})^{2},\label{eq:taylor}
\end{equation}
where $\Delta(S_{0})=\frac{\partial\pi(S_{0})}{\partial S_{0}}$and
$\Gamma(S_{0})=\frac{\partial^{2}\pi(S_{0})}{\partial S_{0}^{2}}$
. We ignore any time effects, since the maturity is far enough in
the future. One can also consider the Delta-Gamma-Theta VaR, see \citet{castellacci2003practice},
where Theta measures the sensitivity of the option price to the passage
of time.

To obtain the Delta-Gamma VaR, compute the Greeks Delta and Gamma
once and simulate $N$ scenarios for $S_{t}$. The Delta-Gamma VaR
is defined by the empirical $0.01$-quantile of the $N$ values for
the portfolio obtained from the right-hand side of Equation (\ref{eq:taylor}).
It takes approximately 3.7 hours to compute the full Monte Carlo VaR
with $N=10^{6}$ using the Carr-Madan formula for pricing. We implement
the full Monte Carlo VaR in the software R with vectorized code on
an AMD Ryzen 5 5500U 2.1 GHz processor. In comparison, the Delta-Gamma
VaR can be obtained in as little as 0.2 seconds by estimating Delta
and Gamma (just once) by finite differences.

To simulate the stock price in the ME model we set $S_{0}=0.75$ and
use the same parameters as in (\ref{eq:paramME}). We use the inverse
transformation, since the quantile function is given in closed form,
see Appendix \ref{sec:ME}. 

To simulate the stock price in the VG model we set $S_{0}=0.65$ and
use the same parameters as in (\ref{eq:paramVG}). A VG random variable
can be simulated by the difference of two independent random variables
following a Gamma distribution, see \citet[Sec. 8.4.2]{schoutens2003levy}.
As can be seen from Table \ref{tab:ME-model} and Table \ref{tab:VG-model},
the Delta-Gamma VaR based on finite differences is close to the full
Monte Carlo VaR. However, the Delta-Gamma VaR based on any of the
three Fourier pricing techniques underestimates the full Monte Carlo
VaR by approximately a factor of 3 for the ME model and by approximately
a factor of 20 to 2000 for the VG model.

\begin{table}[H]
\begin{centering}
\begin{tabular}{|>{\centering}p{2cm}|>{\centering}p{1.3cm}|>{\centering}p{1.3cm}|>{\centering}p{1.3cm}|>{\centering}p{1.3cm}|>{\centering}p{1.3cm}|>{\centering}p{1.3cm}|>{\centering}p{1.3cm}|}
\hline 
 & Analy-tical & finite diff., $h=0.01$ & finite diff., $h=0.02$ & finite diff., $h=0.005$ & COS & Carr-Madan & Lewis\tabularnewline
\hline 
\hline 
$\pi(S_{0})$ & 0.45 & / & / & / & 0.45 & 0.46 & 0.45\tabularnewline
$\Delta(S_{0})$ & -2.10 & -2.10 & -2.10 & -2.10 & -2.09 & -2.09 & -2.01\tabularnewline
$\Gamma(S_{0})$ & 12.47 & 12.47 & 12.55 & 12.46 & 36.98 & 29.01 & 38.54\tabularnewline
\centering{}$\Delta$-$\Gamma$-VaR & \centering{}0.15 & \centering{}0.15 & \centering{}0.15 & \centering{}0.15 & \centering{}0.06 & \centering{}0.08 & \centering{}0.05\tabularnewline
\hline 
\end{tabular}
\par\end{centering}
\caption{\protect\label{tab:ME-model}Delta-Gamma VaR in the ME model: The
portfolio consists of a digital put option and a loan in the risk-free
asset. A full Monte Carlo VaR is equal to $0.18$.}
\end{table}
\begin{table}[H]
\begin{centering}
\begin{tabular}{|>{\centering}p{2cm}|>{\centering}p{1.3cm}|>{\centering}p{1.3cm}|>{\centering}p{1.3cm}|>{\centering}p{1.3cm}|>{\centering}p{1.3cm}|>{\centering}p{1.3cm}|}
\hline 
 & finite diff., $h=0.01$ & finite diff., $h=0.02$ & finite diff., $h=0.005$ & COS & Carr-Madan & Lewis\tabularnewline
\hline 
\hline 
$\pi(S_{0})$ & / & / & / & 0.99 & 0.99 & 0.98\tabularnewline
$\Delta(S_{0})$ & -0.21 & -0.21 & -0.21 & 0.06 & 0.05 & 0.24\tabularnewline
$\Gamma(S_{0})$ & -6.45 & -7.01 & -6.62 & 1347.08 & 1647.57 & 217.87\tabularnewline
$\Delta$-$\Gamma$-VaR & $2.5\times10^{-3}$ & $2.6\times10^{-3}$ & $2.5\times10^{-3}$ & $1.0\times10^{-6}$ & $1.0\times10^{-6}$ & $1.0\times10^{-4}$\tabularnewline
\hline 
\end{tabular}
\par\end{centering}
\caption{\protect\label{tab:VG-model}Delta-Gamma VaR in the VG model: The
portfolio consists of a digital put option and a loan in the risk-free
asset. A full Monte Carlo VaR is equal to $2.1\times10^{-3}$.}
\end{table}

\section{\protect\label{sec:Conclusion}Conclusion}

In this article, we have discussed the approximation of the Greeks
Delta and Gamma of digital put options and plain vanilla call options
in the very simple ME model and the more complex and widely used VG
model. In both models, well-known Fourier pricing techniques such
as the Carr-Madan formula, the COS method or the Lewis method fail
to approximate the Gamma in any meaningful way for some realistic
market parameters because the sufficient conditions for using these
Fourier techniques to obtain Gamma are not satisfied. This can lead
to significant problems, such as a significant underestimation of
VaR when approximated by the Delta-Gamma approach, thereby misrepresenting
the true risk exposure. This certainly shows that one has to be very
careful: One should not blindly trust the output of a Fourier method
without first checking the conditions for its application. A simple
solution is to use Fourier methods only for pricing and to use finite
differences to obtain the Greeks.

\appendix

\section{\protect\label{sec:ME}Technical details}
\begin{prop}
Let $X\sim\text{ME}(\mu\text{, }\eta\text{, }\lambda)$. Then $X$
has the cumulative distribution function
\[
F(x)=\begin{cases}
\frac{1}{2}e^{\eta(x-\mu)} & ,x\leq\mu\\
1-\frac{1}{2}e^{-\lambda(x-\mu)} & ,x\geq\mu\text{.}
\end{cases}
\]
The quantile function is given by 
\[
F^{-1}(y)=\begin{cases}
\frac{1}{\eta}\log(2y)+\mu & ,0<y<0.5\\
-\frac{1}{\lambda}\log(2(1-y))+\mu & ,0.5\leq y<1.
\end{cases}
\]
The characteristic function of $X$ is given by 
\begin{align}
\varphi_{X}(u)= & \frac{e^{iu\mu}}{2}\left(\frac{\lambda}{\lambda-iu}+\frac{\eta}{iu+\eta}\right).\label{eq:phiX}
\end{align}
\end{prop}

\begin{proof}
Integrate the density $f$ in Equation (\ref{eq:f}) to obtain $F$.
Invert $F$ to obtain $F^{-1}$. The characteristic function of $X$
is given by 
\[
\varphi_{X}(u)=E\left[e^{iuX}\right]=\frac{\eta}{2}\left[\frac{1}{iu+\eta}e^{(iu+\eta)x-\eta\mu}\right]_{-\infty}^{\mu}+\frac{\lambda}{2}\left[\frac{1}{iu-\lambda}e^{(iu-\lambda)x+\lambda\mu}\right]_{\mu}^{\infty},
\]
which leads to Equation (\ref{eq:phiX}).
\end{proof}
\begin{prop}
In the ME model, we have $m=-\log\big(\frac{1}{2}\big(\frac{\lambda}{\lambda-\sqrt{T}}+\frac{\eta}{\eta+\sqrt{T}}\big)\big)$.
The characteristic function $\varphi$ of $\log(\bar{S}_{T})$ in
the ME model is given by 
\[
\varphi(u)=\exp\left(iu\big(rT+m\big)\right)\frac{1}{2}\bigg(\frac{\lambda}{\lambda-\sqrt{T}iu}+\frac{\eta}{\eta+\sqrt{T}iu}\bigg).
\]
\end{prop}

\begin{proof}
In the ME model, $X_{T}$ is ME(0, $\frac{\eta}{\sqrt{T}}$, $\frac{\lambda}{\sqrt{T}}$)
distributed. The characteristic function of $X_{T}$ is denoted by
$\varphi_{X_{T}}$. By \citet[Sec. 6.2.2]{schoutens2003levy}, the
mean-correcting term is given by $m=-\log\left(\varphi_{X_{T}}(-i)\right)$.
The characteristic function of $\log(\bar{S}_{T})$ is then given
by $u\mapsto\exp\left(iu\big(rT+m\big)\right)\varphi_{X_{T}}(u)$.
Use Equation (\ref{eq:phiX}) to conclude.
\end{proof}
Let $S_{0}^{\ast}:=K\exp(-m-rT)$ and $d:=\log\left(\frac{K}{S_{0}}\right)-rT-m$.
If $S_{0}=S_{0}^{\ast}$ then $d=0$.
\begin{prop}
The price of a digital put option in the ME model is given by 
\[
\pi_{\text{digital}}(S_{0})=\begin{cases}
\frac{1}{2}\exp\big(-rT+\frac{\eta}{\sqrt{T}}d\big) & ,S_{0}\geq S_{0}^{\ast}\\
\exp(-rT)\big(1-\frac{1}{2}\exp(-\frac{\lambda}{\sqrt{T}}d)\big) & ,\text{otherwise.}
\end{cases}
\]
\end{prop}

\begin{proof}
The arbitrage-free price of a digital put option is given by 
\begin{align*}
\pi_{\text{digital}}(S_{0})= & e^{-rT}E\left[1_{\{S_{T}<K\}}\right]\\
= & e^{-rT}Q\left[X_{T}\leq\log\left(\frac{K}{S_{0}}\right)-rT-m\right]\\
= & e^{-rT}F_{X_{T}}\left(d\right)\text{.}
\end{align*}
In the ME model, $X_{T}$ is ME(0, $\frac{\eta}{\sqrt{T}}$, $\frac{\lambda}{\sqrt{T}}$)
distributed, which concludes the proof.
\end{proof}
\begin{prop}
The Delta of a digital put option in the ME model is given by 
\[
\Delta_{\text{digital}}(S_{0})=\frac{\partial\pi_{\text{digital}}(S_{0})}{\partial S_{0}}=\begin{cases}
-\frac{\eta}{2\sqrt{T}S_{0}}\exp\big(-rT+\frac{\eta}{\sqrt{T}}d\big) & ,S_{0}>S_{0}^{\ast}\\
-\frac{\lambda}{2\sqrt{T}S_{0}}\exp\big(-rT-\frac{\lambda}{\sqrt{T}}d\big) & ,S_{0}<S_{0}^{\ast}.
\end{cases}
\]
Delta jumps at $S_{0}^{\ast}$ if $\lambda\neq\eta$.
\end{prop}

\begin{proof}
For $S_{0}>S_{0}^{*}$ it holds that 
\begin{align*}
\frac{\partial\pi_{\text{digital}}(S_{0})}{\partial S_{0}}= & \frac{\partial}{\partial S_{0}}\frac{1}{2}\exp\big(-rT+\frac{\eta}{\sqrt{T}}d\big)\\
= & -\frac{\eta}{2\sqrt{T}S_{0}}\exp\big(-rT+\frac{\eta}{\sqrt{T}}d\big)\text{.}
\end{align*}
For $S_{0}<S_{0}^{*}$ it holds that 
\begin{align*}
\frac{\partial\pi_{\text{digital}}(S_{0})}{\partial S_{0}}= & \frac{\partial}{\partial S_{0}}\exp(-rT)\big(1-\frac{1}{2}\exp(-\frac{\lambda}{\sqrt{T}}d)\big)\\
= & -\frac{\lambda}{2\sqrt{T}S_{0}}\exp\big(-rT-\frac{\lambda}{\sqrt{T}}d\big)\text{.}
\end{align*}
For $S_{0}=S_{0}^{*}$ Delta jumps because for $\lambda\neq\eta$
it holds that
\[
-\frac{e^{-rT}}{2\sqrt{T}S_{0}}\eta\neq-\frac{e^{-rT}}{2\sqrt{T}S_{0}}\lambda\text{.}
\]
\end{proof}
\begin{prop}
The Gamma of a digital put option is undefined at $S_{0}^{\ast}$
if $\lambda\neq\eta$. For $S_{0}\neq S_{0}^{\ast}$ the Gamma as
a pointwise derivative is given by 
\[
\Gamma_{\text{digital}}(S_{0})=\frac{\partial^{2}\pi_{\text{digital}}(S_{0})}{\partial S_{0}^{2}}=\begin{cases}
\frac{1}{2S_{0}^{2}}\left(\frac{\eta^{2}}{T}+\frac{\eta}{\sqrt{T}}\right)\exp\big(-rT+\frac{\eta}{\sqrt{T}}d\big) & ,S_{0}>S_{0}^{\ast}\\
\frac{1}{2S_{0}^{2}}\left(-\frac{\lambda^{2}}{T}+\frac{\lambda}{\sqrt{T}}\right)\exp\big(-rT-\frac{\lambda}{\sqrt{T}}d\big) & ,S_{0}<S_{0}^{\ast}.
\end{cases}
\]
\end{prop}

\begin{proof}
The Gamma of a digital put option is undefined at $S_{0}^{\ast}$
since the Delta jumps at $S_{0}^{\ast}$. For $S_{0}>S_{0}^{*}$ it
holds that 
\begin{align*}
\frac{\partial^{2}\pi_{\text{digital}}(S_{0})}{\partial S_{0}^{2}}= & \frac{\partial}{\partial S_{0}}\left(-\frac{\eta}{2\sqrt{T}S_{0}}\exp\big(-rT+\frac{\eta}{\sqrt{T}}d\big)\right)\\
= & -\frac{\eta\exp\big(-rT+\frac{\eta}{\sqrt{T}}d\big)\biggl(-\frac{\eta}{\sqrt{T}S_{0}}\biggr)}{2\sqrt{T}S_{0}}+\frac{\eta\exp\big(-rT+\frac{\eta}{\sqrt{T}}d\big)}{2\sqrt{T}S_{0}^{2}}\\
= & \frac{1}{2S_{0}^{2}}\left(\frac{\eta^{2}}{T}+\frac{\eta}{\sqrt{T}}\right)\exp\big(-rT+\frac{\eta}{\sqrt{T}}d\big)\text{.}
\end{align*}
For $S_{0}<S_{0}^{*}$ it holds that 
\begin{align*}
\frac{\partial^{2}\pi_{\text{digital}}(S_{0})}{\partial S_{0}^{2}}= & \frac{\partial}{\partial S_{0}}\left(-\frac{\lambda}{2\sqrt{T}S_{0}}\exp\big(-rT-\frac{\lambda}{\sqrt{T}}d\big)\right)\\
= & -\frac{\lambda}{2\sqrt{T}S_{0}}\exp\big(-rT-\frac{\lambda}{\sqrt{T}}d\big)\frac{\lambda}{\sqrt{T}S_{0}}+\frac{\lambda}{2\sqrt{T}S_{0}^{2}}\exp\big(-rT-\frac{\lambda}{\sqrt{T}}d\big)\\
= & \frac{1}{2S_{0}^{2}}\left(-\frac{\lambda^{2}}{T}+\frac{\lambda}{\sqrt{T}}\right)\exp\big(-rT-\frac{\lambda}{\sqrt{T}}d\big)\text{.}
\end{align*}
\end{proof}

\begin{lem}
\label{lem:Grubb-CM} Let $M\subset\mathbb{R}$ be measurable, e.g.,
$M=(0,\infty)$. Let $h(v,S_{0})$ be a family of functions of $v\in M$
depending on the parameter $S_{0}>0$, such that for any $S_{0}>0$,
the map $v\mapsto h(v,S_{0})$ is integrable. Consider the function
$\pi$ defined by 
\[
\pi(S_{0})=\int_{M}h(v,S_{0})dv,\quad S_{0}>0.
\]
Assume that $\frac{\partial}{\partial S_{0}}h(v,S_{0})$ exists for
all $(v,S_{0})\in M\times(0,\infty),$ and there is an integrable
function $g$ such that 
\[
\left|\frac{\partial}{\partial S_{0}}h(v,S_{0})\right|\leq g(v)
\]
for all $(v,S_{0})\in M\times(0,\infty)$. Then $S_{0}\mapsto\pi(S_{0})$
is a differentiable function and 
\[
\frac{d}{dS_{0}}\pi(S_{0})=\text{\ensuremath{\int_{M}}}\frac{\partial}{\partial S_{0}}h(v,S_{0})dv\text{.}
\]
\end{lem}

\begin{proof}
\citet[Lemma 2.8]{grubb2008distributions}.
\end{proof}
\bibliographystyle{plainnat}
\bibliography{paper_Literaturverzeichnis}

\begin{thebibliography}{29}
\providecommand{\natexlab}[1]{#1}
\providecommand{\url}[1]{\texttt{#1}}
\expandafter\ifx\csname urlstyle\endcsname\relax
  \providecommand{\doi}[1]{doi: #1}\else
  \providecommand{\doi}{doi: \begingroup \urlstyle{rm}\Url}\fi

\bibitem[Bayer et~al.(2023)Bayer, Hammouda, Papapantoleon, Samet, and
  Tempone]{bayer2022optimal}
C.~Bayer, C.~B. Hammouda, A.~Papapantoleon, M.~Samet, and R.~Tempone.
\newblock {Optimal Damping with Hierarchical Adaptive Quadrature for Efficient
  Fourier Pricing of Multi-Asset Options in L{\'e}vy Models}.
\newblock \emph{Journal of Computational Finance}, 27\penalty0 (3):\penalty0
  43--86, 2023.

\bibitem[Carr and Madan(1999)]{carr1999option}
Peter Carr and Dilip Madan.
\newblock Option valuation using the fast fourier transform.
\newblock \emph{Journal of computational finance}, 2\penalty0 (4):\penalty0
  61--73, 1999.

\bibitem[Castellacci and Siclari(2003)]{castellacci2003practice}
Giuseppe Castellacci and Michael~J Siclari.
\newblock The practice of delta--gamma var: Implementing the quadratic
  portfolio model.
\newblock \emph{European Journal of Operational Research}, 150\penalty0
  (3):\penalty0 529--545, 2003.

\bibitem[Corcuera et~al.(2009)Corcuera, Guillaume, Leoni, and
  Schoutens]{corcuera2009implied}
Jos{\'e}~Manuel Corcuera, Florence Guillaume, Peter Leoni, and Wim Schoutens.
\newblock Implied l{\'e}vy volatility.
\newblock \emph{Quantitative Finance}, 9\penalty0 (4):\penalty0 383--393, 2009.

\bibitem[Duffie and Pan(2001)]{duffie2001analytical}
Darrell Duffie and Jun Pan.
\newblock {Analytical value-at-risk with jumps and credit risk}.
\newblock \emph{Finance and Stochastics}, 5:\penalty0 155--180, 2001.

\bibitem[Eberlein et~al.(2010)Eberlein, Glau, and
  Papapantoleon]{eberlein2010analysis}
Ernst Eberlein, Kathrin Glau, and Antonis Papapantoleon.
\newblock {Analysis of Fourier transform valuation formulas and applications}.
\newblock \emph{Applied Mathematical Finance}, 17\penalty0 (3):\penalty0
  211--240, 2010.

\bibitem[Eberlein et~al.(2016)Eberlein, Eddahbi, and Lalaoui
  Ben~Cherif]{eberlein2016computation}
Ernst Eberlein, M'hamed Eddahbi, and SM~Lalaoui Ben~Cherif.
\newblock {Computation of Greeks in LIBOR models driven by time--inhomogeneous
  L{\'e}vy processes}.
\newblock \emph{Applied Mathematical Finance}, 23\penalty0 (3):\penalty0
  236--260, 2016.

\bibitem[Fang and Oosterlee(2009a)]{fang2009novel}
Fang Fang and Cornelis~W Oosterlee.
\newblock {A novel pricing method for European options based on Fourier-cosine
  series expansions}.
\newblock \emph{SIAM Journal on Scientific Computing}, 31\penalty0
  (2):\penalty0 826--848, 2009a.

\bibitem[F{\"o}llmer and Schied(2011)]{follmer2011stochastic}
Hans F{\"o}llmer and Alexander Schied.
\newblock \emph{{Stochastic finance: an introduction in discrete time}}.
\newblock Walter de Gruyter, 2011.

\bibitem[Glasserman(2004)]{glasserman2004monte}
Paul Glasserman.
\newblock \emph{Monte Carlo methods in financial engineering}, volume~53.
\newblock Springer, 2004.

\bibitem[Grubb(2008)]{grubb2008distributions}
Gerd Grubb.
\newblock \emph{{Distributions and operators}}, volume 252.
\newblock Springer Science \& Business Media, 2008.

\bibitem[Heston(1993)]{heston1993closed}
S.~L. Heston.
\newblock {A closed-form solution for options with stochastic volatility with
  applications to bond and currency options}.
\newblock \emph{The Review of Financial Studies}, 6\penalty0 (2):\penalty0
  327--343, 1993.

\bibitem[Hull(2021)]{HullJohn2021Ofao}
John Hull.
\newblock \emph{{Options, futures, and other derivatives}}.
\newblock Eleventh edition, global edition edition, 2021.
\newblock ISBN 9781292410623.

\bibitem[Junike(2024)]{junike2024number}
Gero Junike.
\newblock {On the number of terms in the COS method for European option
  pricing}.
\newblock \emph{Numerische Mathematik}, 156\penalty0 (2):\penalty0 533--564,
  2024.

\bibitem[Junike and Pankrashkin(2022)]{junike2022precise}
Gero Junike and Konstantin Pankrashkin.
\newblock {Precise option pricing by the COS method--How to choose the
  truncation range}.
\newblock \emph{Applied Mathematics and Computation}, 421:\penalty0 126935,
  2022.

\bibitem[K{\"u}chler and Tappe(2008)]{kuchler2008shapes}
Uwe K{\"u}chler and Stefan Tappe.
\newblock {On the shapes of bilateral Gamma densities}.
\newblock \emph{Statistics \& Probability Letters}, 78\penalty0 (15):\penalty0
  2478--2484, 2008.

\bibitem[Leitao et~al.(2018)Leitao, Oosterlee, Ortiz-Gracia, and
  Bohte]{leitao2018data}
{\'A}lvaro Leitao, Cornelis~W Oosterlee, Luis Ortiz-Gracia, and Sander~M Bohte.
\newblock {On the data-driven COS method}.
\newblock \emph{Applied Mathematics and Computation}, 317:\penalty0 68--84,
  2018.

\bibitem[Lewis(2001)]{lewis2001simple}
Alan~L Lewis.
\newblock {A simple option formula for general jump-diffusion and other
  exponential L{\'e}vy processes}.
\newblock \emph{Available at SSRN 282110}, 2001.

\bibitem[Lord et~al.(2008)Lord, Fang, Bervoets, and Oosterlee]{lord2008fast}
R.~Lord, F.~Fang, F.~Bervoets, and C.~W. Oosterlee.
\newblock {A fast and accurate FFT-based method for pricing early-exercise
  options under L{\'e}vy processes}.
\newblock \emph{SIAM Journal on Scientific Computing}, 30\penalty0
  (4):\penalty0 1678--1705, 2008.

\bibitem[Madan and Schoutens(2016)]{madan2016applied}
Dilip Madan and Wim Schoutens.
\newblock \emph{{Applied conic finance}}.
\newblock Cambridge University Press, 2016.

\bibitem[Madan(2016)]{madan2016adapted}
Dilip~B Madan.
\newblock {Adapted hedging}.
\newblock \emph{Annals of Finance}, 12\penalty0 (3):\penalty0 305--334, 2016.

\bibitem[Madan et~al.(1998)Madan, Carr, and Chang]{madan1998variance}
Dilip~B Madan, Peter~P Carr, and Eric~C Chang.
\newblock {The variance gamma process and option pricing}.
\newblock \emph{Review of Finance}, 2\penalty0 (1):\penalty0 79--105, 1998.

\bibitem[Martins and Ning(2021)]{martins2021engineering}
Joaquim~RRA Martins and Andrew Ning.
\newblock \emph{{Engineering design optimization}}.
\newblock Cambridge University Press, 2021.

\bibitem[Ortiz-Gracia and Oosterlee(2013)]{ortiz2013robust}
L.~Ortiz-Gracia and C.~W. Oosterlee.
\newblock {Robust pricing of European options with wavelets and the
  characteristic function}.
\newblock \emph{SIAM Journal on Scientific Computing}, 35\penalty0
  (5):\penalty0 B1055--B1084, 2013.

\bibitem[Ortiz-Gracia and Oosterlee(2016)]{ortiz2016highly}
L.~Ortiz-Gracia and C.~W. Oosterlee.
\newblock {A highly efficient Shannon wavelet inverse Fourier technique for
  pricing European options}.
\newblock \emph{SIAM Journal on Scientific Computing}, 38\penalty0
  (1):\penalty0 B118--B143, 2016.

\bibitem[Owolabi and Atangana(2019)]{OwolabiKoladeM2019Nmff}
Kolade~M Owolabi and Abdon Atangana.
\newblock \emph{{Numerical methods for fractional differentiation}}.
\newblock Springer eBooks Mathematics and Statistics. 2019.
\newblock ISBN 9789811500985.

\bibitem[Sauer(2013)]{sauer2013}
Timothy Sauer.
\newblock \emph{{Numerical Analysis}}.
\newblock Pearson Deutschland, 2 edition, 2013.
\newblock ISBN 9781292023588.

\bibitem[Schoutens(2003)]{schoutens2003levy}
W.~Schoutens.
\newblock \emph{{L{\'e}vy Processes in Finance: Pricing Financial
  Derivatives}}.
\newblock Wiley Online Library, 2003.

\bibitem[Takahashi and Yamazaki(2009)]{takahashi2009efficient}
Akihiko Takahashi and Akira Yamazaki.
\newblock {Efficient static replication of European options under exponential
  L{\'e}vy models}.
\newblock \emph{Journal of Futures Markets: Futures, Options, and Other
  Derivative Products}, 29\penalty0 (1):\penalty0 1--15, 2009.

\end{thebibliography}

\end{document}